\documentclass[a4paper]{article}

\usepackage{graphicx}
\usepackage{amsmath, amsfonts, amssymb, amsthm}
\usepackage{algorithm}
\usepackage{algorithmic}
\usepackage{color}

\newtheorem{theorem}{Theorem}[section]
\newtheorem{lemma}[theorem]{Lemma}
\newtheorem{proposition}[theorem]{Proposition}
\newtheorem{corollary}[theorem]{Corollary}

\theoremstyle{definition}
\newtheorem{definition}[theorem]{Definition}
\newtheorem{example}[theorem]{Example}
\newtheorem{remark}[theorem]{Remark}

\newcommand{\fq}{\mathbb{F}_q}

\title{List Decoding of Matrix-Product Codes\\from nested codes: an application to\\Quasi-Cyclic codes\thanks{This work was supported in part by the Danish National Research Foundation and the National Science Foundation of China (Grant No.11061130539) for the Danish-Chinese Center for Applications of Algebraic Geometry in Coding Theory and Cryptography, by the Claude Shannon Institute, Science Foundation Ireland Grant 06/MI/006 and by Spanish MEC Grant MTM2007-64704. \textit{2000 Mathematics Subject Classification}. Primary: 94B05; Secondary: 94B35. \textit{Keywords:} Linear Code, Matrix-Product Code, List Decoding, Quasi-Cyclic Code.}}

\date{}

\begin{document}

\maketitle

\centerline{\scshape Fernando Hernando\footnote{INSPIRE fellow funding received from the Irish Research Council for Science, Engineering and Technology.}}
\medskip
{\footnotesize
 \centerline{Department of Mathematics, Universidad Jaume I,}
   \centerline{Campus Riu Sec, 12071, Castell\'on de la plana,
   Spain}
   \centerline{\texttt{carrillf@mat.uji.es}}
}

\medskip

\centerline{\scshape Tom H{\o}holdt
}
\medskip
{\footnotesize
 \centerline{DTU-Mathematics, Technical University of Denmark,}
   \centerline{Matematiktorvet, Building 303, 2800 Kgs. Lyngby, Denmark}
   \centerline{\texttt{T.Hoeholdt@mat.dtu.dk}}
}

\medskip

\centerline{\scshape Diego Ruano
}
\medskip
{\footnotesize
 \centerline{Department of Mathematical Sciences, Aalborg University,}
   \centerline{Fr. Bajersvej 7G, 9920-Aalborg {\O}st, Denmark}
   \centerline{\texttt{diego@math.aau.dk}}
}

\medskip

\begin{abstract}
A list decoding algorithm for matrix-product codes is provided when $C_1, \ldots , C_s$ are nested linear codes and $A$ is a non-singular by columns matrix. We estimate the probability of getting more than one codeword as output when the constituent codes are Reed-Solomon codes. We extend this list decoding algorithm for matrix-product codes with polynomial units, which are quasi-cyclic codes. Furthermore, it allows us to consider unique decoding for matrix-product codes with polynomial units. 
\end{abstract}

\section{Introduction}
Matrix-product codes, $[C_{1}\cdots C_{s}]\cdot A$, are a generalization of several classic codes constructions of codes from old ones \cite{Blackmore-Norton,Ozbudak}. For instance, they extend the $(u|u+v)$-construction. An algorithm for unique decoding  when the codes are nested, $C_1 \supset \cdots \supset C_s$, and $A$ has a certain property, called non-singular by columns, was provided in \cite{hlr}. The algorithm decodes up to half of the minimum distance, assuming that we have a decoding algorithm for $C_i$ that decodes up to half of its minimum distance, for every $i$. 

List decoding was introduced by Elias \cite{eli} and Wozencraft
\cite{woz}. The list decoder is a relaxation over unique decoding that allows the decoder to produce a
list of codewords as answers. It can uniquely decode beyond half of
the minimum distance in some cases or to produces a list of
codewords.

In $1997$,  Sudan presented a polynomial time
algorithm for decoding low rate Reed-Solomon codes beyond the classical
$\frac{d}{2}$ bound. Later \cite{Guruswami-Sudan}, Guruswami and
 Sudan provided a significantly improved version of list
decoder which can correct codes of any rates. Recently, Lee and O'Sullivan  provide a list decoding
algorithm based on the computation of a Gr\"{o}bner basis of a module \cite{Lee-Michael} and Beelen and Brander provide an algorithm that has linear complexity in the code length \cite{peter}.

In this paper we consider a list decoding algorithm for matrix-product codes which is an extension of the algorithm in
\cite{hlr}. The algorithm in \cite{hlr} assumes a known decoding algorithm for every constituent code $C_i$ that decodes up to half of its minimum distance, for this algorithm, we assume that the decoding algorithm is a list-decoding algorithm. Moreover, it is also required that $C_1, \ldots , C_s$ are nested and $A$ is non-singular by columns. The extension is natural, but, we believe, it is non-trivial task since we had to modify the algorithm to deal with lists of codewords, compute the error bound $\tau$ and prove the correctness of the algorithm, among others.

Matrix-Product codes are generalized concatenated codes \cite{Blackmore-Norton, Dumer}, which have an efficient decoding algorithm \cite{Guruswami-Rudra}. However, this algorithm cannot be successfully applied if the matrix $A$ is small, as it is in practice for Matrix-product codes (see Remark \ref{re:GR}).

The probability of getting more than one codeword as output of a list decoding algorithm for Reed-Solomon codes was bounded in \cite{Tom}. In section \ref{se:rs}, we use this computation to estimate an upper bound of the probability of getting more than one codeword as output, when $C_1, \ldots , C_s$ are Reed-Solomon codes. The algorithm in section \ref{BigTau} can become computationally intense, an optimal situation arises considering $s=l=2$.

In section \ref{se:units} we extend the algorithm in section \ref{BigTau} for matrix-product codes with polynomial units \cite{hr}, which are quasi-cyclic codes. Quasi-cyclic codes became important after it was shown that
some codes in this class meet a modified Gilbert-Varshamov bound  \cite{Kas}, however there are no general fast algorithms for decoding them. In \cite{hr}, many of these codes with good parameters were obtained. Using list decoding of matrix-product codes with polynomial units we can uniquely decode these codes up to the half of the minimum distance.

\section{Matrix-Product Codes}\label{sect:mpc}

A matrix-product code is a construction of a code from old ones.

\begin{definition}
Let $C_1, \ldots, C_s \subset \mathbb{F}_q^m$ be linear codes of length
$m$ and a matrix $A=(a_{i,j}) \in \mathcal{M}(\fq, s \times l)$, with
$s\leq l$. The \textbf{matrix-product code} $C=[C_1 \cdots C_s] \cdot A$
is the set of all matrix-products $[c_1 \cdots c_s] \cdot A$ where $c_i\in
C_i$ is an $m \times 1$ column vector $c_i=(c_{1,i},\ldots,c_{m,i})^T$ for $i=1,\ldots, s$. Therefore, a typical codeword $\mathbf{c}$ is

\begin{equation}\label{MatrixCodeword}
\mathbf{c}= \left(
\begin{tabular}{ccc}
$c_{1,1}a_{1,1}+\cdots+c_{1,s}a_{s,1}$ & $\cdots$ & $c_{1,1}a_{1,l}+\cdots+c_{1,s}a_{s,l}$\\
$\vdots$ & $\ddots$ & $\vdots$\\
$c_{m,1}a_{1,1}+\cdots+c_{m,s}a_{s,1}$ & $\cdots$ & $c_{m,1}a_{1,l}+\cdots+c_{m,s}a_{s,l}$\\
\end{tabular}\right).
\end{equation}
\end{definition}

Clearly the $i$-th column of any codeword is an element of the form $\sum_{j=1}^s a_{j,i}
c_j\in \mathbb{F}_q^m$, therefore reading the entries of the $m\times l$-matrix above in column-major order, the codewords can be viewed as vectors of length $ml$,
\begin{equation}\label{VectorCodeword}
\mathbf{c}=\left(\sum_{j=1}^s a_{j,1} c_j, \ldots , \sum_{j=1}^s a_{j,l} c_j \right)
\in\mathbb{F}_q^{ml}.
\end{equation}

From the above construction it follows that a generator matrix
of $C$ is of the form:
\begin{displaymath}
G=\left(
\begin{tabular}{cccccc}
$a_{1,1}G_1$ & $a_{1,2}G_1$ & $\cdots$ & $a_{1,s}G_1$& $\cdots$ & $a_{1,l}G_1$\\
$a_{2,1}G_2$ & $a_{2,2}G_2$& $\cdots$ & $a_{2,s}G_2$ & $\cdots$ & $a_{2,l}G_2$\\
$\vdots$ & $\vdots$& $\cdots$ & $\vdots$& $\cdots$ & $\vdots$\\
$a_{s,1}G_s$ & $a_{s,2}G_s$& $\cdots$ & $a_{s,s}G_s$ & $\cdots$ & $a_{s,l}G_s$\\
\end{tabular}\right),
\end{displaymath}
where $G_i$ is a generator matrix of $C_i$, $i=1,\ldots,s$.
Moreover, if $C_i$ is a $[m,k_i,d_i]$ code then one has that
$[C_1 \cdots C_s] \cdot A$ is a linear code over $\mathbb{F}_q$ with
length $lm$ and dimension $k=k_1+\cdots+k_s$  if the matrix $A$ has
full rank and $k < k_1+\cdots+k_s$ otherwise.

Let us denote by $R_i= (a_{i,1},\ldots,a_{i,l})$ the element of
$\mathbb{F}_q^l$ consisting of the $i$-th row of $A$, for
$i=1,\ldots,s$. We denote by $D_i$ the minimum distance
of the code $C_{R_i}$ generated by $\langle R_1,\ldots,
R_i\rangle$ in $\fq^l$. In \cite{Ozbudak} the following lower bound for the minimum distance of the matrix-product code $C$ is obtained,
\begin{equation}\label{lowebound}
d(C)\geq \min\{d_1D_1,d_2D_2, \ldots ,d_s D_s\},
\end{equation}where $d_i$ is the minimum distance of $C_i$. If $C_1, \ldots, C_s$ are nested codes, $C_1 \supset \cdots \supset C_s$, the previous bound is sharp \cite{hlr}.

In \cite{Blackmore-Norton}, the following condition for the matrix $A$ is introduced.

\begin{definition}\cite{Blackmore-Norton}\label{de:nsc}
Let $A$ be a $s\times l$ matrix and  $A_t$ be the matrix consisting of the first $t$ rows of $A$. For $1\leq j_1< \cdots < j_t\leq l$, we denote by $A(j_1,\ldots,j_t)$ the $t\times t$ matrix consisting of the columns $j_1,\ldots,j_t$ of $A_t$.

A matrix $A$ is non-singular by columns if $A(j_1,\ldots,j_t)$ is non-singular for each $1\leq t\leq s$ and $1\leq j_1< \cdots < j_t\leq l$. In particular, a non-singular by columns matrix $A$  has full rank.
\end{definition}

Moreover, if $A$ is non-singular by columns and $C_1 \supset \cdots \supset C_s$, we have $d(C) = \min\{ld_1,(l-1)d_2, \ldots ,(l-s+1)d_s \}$ \cite{hlr}.

In \cite{hlr} were presented a decoding algorithm for the matrix-product code $C=[C_1 \cdots C_s] \cdot A \subset \fq^{ml}$, with $A$ non-singular by columns and $C_1 \supset \cdots \supset C_s$, assuming that we have a decoding algorithm for $C_i$, for $i=1,\ldots,s$. The algorithm in \cite{hlr} decodes up to half of the minimum distance. In next section we provide a list decoding algorithm for such codes, assuming that we have a list decoding algorithm for $C_i$, $i=1,\ldots, s$.

\section{List Decoding Algorithm for matrix-product codes}\label{BigTau}

Let $C \subset\mathbb{F}_q^n$ and $\tau> 1$. For $r \in \mathbb{F}_q^n$,
a list decoding algorithm with error bound $\tau$ provides a list with all
codewords in $C$ that differ from $r$ in at most $\tau$ places. If $\tau \leq \lfloor \frac{d-1}{2} \rfloor$, it will result into unique decoding.

We present a list decoding algorithm for a class of matrix-product codes,
it is an extension of \cite[Algorithm 1]{hlr}. Namely, we consider $s$ nested linear codes $C_1,\ldots, C_s \subset \fq^m$ and a non-singular by columns matrix $A \in \mathcal{M}(\fq,s \times l)$, where $s \le l$. We provide a list decoding algorithm for the matrix-product code $C=[C_1 \cdots C_s] \cdot A \subset \fq^{ml}$, assuming that we have a list decoding algorithm $LDC_i$ for $C_i$ with error bound $\tau_i$. In particular, each $LDC_i$  answers an empty list if there is no
codeword in $C_i$ within distance $\tau_i$  of the received word.

Our list algorithm for $C$ decodes up to
\begin{equation}\label{TauValue}
\tau=\min \{l\tau_1+(l-1),(l-1)\tau_2+(l-2), \ldots, (l-s+1)\tau_s+l-s \}.
\end{equation}

We first describe the main steps in our decoding algorithm. The algorithm is outlined as a whole in procedural form in Algorithm
\ref{alg:dec}.

Consider the codeword $\mathbf{c}=(\sum_{j=1}^s a_{j,1} c_j, \ldots , \sum_{j=1}^s a_{j,l} c_j )$, where $c_j \in C_j$, for all $j$. Suppose that $\mathbf{c}$ is sent and that we receive $\mathbf{p} = \mathbf{c} + \mathbf{e}$, where $\mathbf{e}=(e_1,e_2,\ldots,e_l)\in \fq^{ml}$ is an error vector. We denote  by $p_i = \sum_{j=1}^s a_{j,i} c_j + e_i \in \fq^{m}$ the $i$-th block of $\mathbf{p}$, for $i=1,\ldots,l$. Let $\{ i_1 , \ldots , i_s \} \subset \{ 1,\ldots , l  \}$ be an ordered subset of indices. We now also suppose that $\mathbf{e}$ satisfies the extra property that
\begin{equation}\label{eq:condi}
wt(e_{i_j}) \le \tau_j ~ \mathrm{for~all}~j \in \{ 1,\ldots, s\}.
\end{equation}

Since $C_1 \supset \cdots \supset C_s$, each block $\sum_{j=1}^s a_{j,i} c_j$ of $\mathbf{c}$ is a codeword of $C_1$. Therefore, we decode the $i_1$-th block $p_{i_1}$ of $\mathbf{p}$ using $LDC_1$ and we obtain a list $L_1$. Since $wt(e_{i_1}) \le \tau_1$, we have $\sum_{j=1}^s a_{j,i_1} c_j \in L_1$. In practice we do not know which one of the elements in $L_1$ is $\sum_{j=1}^s a_{j,1} c_j$, therefore we should consider the following computations for every element in $L_1$. Assume now that we consider $\sum_{j=1}^s a_{j,i_1} c_j \in L_1$, hence we obtain $e_{i_1}= p_{i_1} - \sum_{j=1}^s a_{j,i_1} c_j$ and we can eliminate $c_1$ in every other block (although we do not know $c_1$) in the following way: we consider a new vector $\mathbf{p}^{2)} \in
\fq^{ml}$ with components
$$
p^{2)}_i = p_i - \frac{a_{1,i}}{a_{1,i_1}} (p_{i_1} - e_{i_1}) =
\sum_{j=2}^s a^{2)}_{j,i} c_j   + e_i, ~ \mathrm{for}~ i \neq i_1,
$$
where
$a^{2)}_{j,i}=a_{j,i}-\frac{a_{1,i}}{a_{1,i_1}}a_{j,i_1}$, and
$p^{2)}_{i_1} = p_{i_1} - e_{i_1}$. Since $A$ is a non-singular by columns matrix, the elements of the first row of $A$ are non-zero, and so the denominator $a_{1,i_1}$ is non-zero.

Since $C_2 \supset \cdots \supset C_s$, we notice that the $i$-th block of $\mathbf{p}^{2)}$ is a codeword of $C_2$ plus the error block $e_i$, for $i \in \{ 1, \ldots, s\} \setminus \{i_1\}$. We now decode the $i_2$-th block $p^{2)}_{i_2} = \sum_{j=2}^s a^{2)}_{j,i_2} c_j  + e_{i_2}$ of $\mathbf{p}^{2)}$ using $LDC_2$ and we obtain a list $L_2$. Since $w(e_{i_2}) \le \tau_2$, we have $\sum_{j=2}^s a^{2)}_{j,i_2} c_j \in L_2$. In practice we do not know again which one of the elements in $L_2$ is $\sum_{j=2}^s a^{2)}_{j,i_2} c_j$, therefore we should consider the following computations for every element in $L_2$. Assume now that we consider $\sum_{j=2}^s a^{2)}_{j,i_2} c_j \in L_2$, hence we obtain $e_{i_2}$ and, as before, we can eliminate $c_2$ in every other block (although we do not know $c_2$) as follows: we consider a new vector $\mathbf{p}^{3)} \in \fq^{ml}$ with components
$$
p^{3)}_i = p^{2)}_i - \frac{a^{2)}_{2,i}}{a^{2)}_{2,i_2}}
(p^{2)}_{i_2} - e_{i_2}) = \sum_{j=3}^s a^{3)}_{j,i} c_j   + e_i,
~ \mathrm{for}~ i \neq i_1,i_2,
$$
where
$a^{3)}_{j,i}=a_{j,i}^{2)}-\frac{a_{2,i}^{2)}}{a_{2,i_2}^{2)}}a_{j,i_2}^{2)}$, $p^{3)}_{i_1} = p^{2)}_{i_1}$ and $p^{3)}_{i_2} = p^{2)}_{i_2} - e_{i_2}$.

Notice that the $i$-th block of $\mathbf{p}^{3)}$ is a codeword of $C_3$
plus the error block $e_i$, for $i \in \{ 1, \ldots, s\} \setminus
\{i_1, i_2\}$.

Then we iterate this process, defining $\mathbf{p}^{k)}$ for
$k=3,\ldots,s$, and decoding the $i_k$-th block using $LDC_k$. In
this way, we obtain the error blocks $e_i$, and the corresponding
codeword blocks $\sum_{j=1}^s a_{j,i} c_j$, for $i \in
\{i_1,\ldots , i_s\}$. The vector $(\sum_{j=1}^s a_{j,i_1} c_j,
\ldots , \sum_{j=1}^s a_{j,i_s} c_j )$ formed from these $s$
decoded blocks is equal to the product $[c_1 \cdots c_s] \cdot
A(i_1,\ldots,i_s)$, where $A(i_1,\ldots,i_s)$ is the $s\times
s$-submatrix of $A$ consisting of the columns $i_1, \ldots, i_s$.
Since this matrix is full rank, we can now easily compute $c_1,
\ldots , c_s$ by inverting $A(i_1,\ldots,i_s)$ or solving the
corresponding linear system. Finally we recover the remaining
$l-s$ codeword blocks ``for free" (i.e. no decoding procedure is
involved for these blocks) by simply recomputing the entire
codeword $\mathbf{c}=[c_1 \cdots c_s] \cdot A=(\sum_{j=1}^s a_{j,1} c_j,
\ldots , \sum_{j=1}^s a_{j,l} c_j )$, since we know the $c_j$'s
and the matrix $A$.

For each elimination step in the above procedure, it is necessary
that $a^{k)}_{k,i_k} \neq 0$, for each $k=2,\ldots,s$, to avoid
zero division. We claim that this follows from the
non-singular by columns property of $A$, exactly in the same way as in \cite{hlr}. Let $A^{1)}=A$. The
matrix $A^{k)} = (a^{k)}_{i,j})\in \mathcal{M}(\fq,s \times
l),k=2,\ldots,s,$ is obtained recursively from $A^{k-1)}$ by
performing the following $l-(k-1)$ elementary column operations:
$$\mathrm{column}_{i}(A^{k)})= \mathrm{column}_{i}(A^{k-1)}) -
\frac{a^{k-1)}_{k-1,i}}{a^{k-1)}_{k-1,i_{k-1}}}\mathrm{column}_{i_{k-1}}(A^{k-1)}),$$
for each $i \notin \{ i_1,\ldots , i_{k-1} \}$. These operations
introduce $l-(k-1)$ additional zero elements in the $k-1$-th row
of $A^{k)}$ at each iteration. Hence the submatrix of $A^{k)}$ given
by the first $k$ rows and the $i_1, \ldots , i_{k}$ columns, is a
triangular matrix (in this case, a column permutation of a lower
triangular matrix) whose determinant is $a^{k)}_{1,i_1} \cdots
a^{k)}_{k,i_{k}}$. Since $A$ is non-singular by columns, this submatrix is non-singular. It follows that the
determinant is non-zero, and therefore $a^{k)}_{k,i_{k}}\ne 0$.

The procedure described above will generate a list that includes the sent word, if for a given choice of indices $\{ i_1 , \ldots , i_s \} \subset \{
1,\ldots , l  \}$, each error block satisfies $wt(e_{i_j})\le \tau_j$, for all $j=1,\ldots,s$. The output's procedure
may not include the sent word if $wt(e_{i_j})> \tau_j$ for some $j$.

In the previous description, we have only shown the computations for one of the different choices that the list decoder algorithms $C_i$, for $i=1, \ldots, s$, may give us. However, if $\# L_1 >1$ then we should consider a different word $\mathbf{p}^{2)}$ for every $\ell \in L_1$. One can see how this tree is created for every element in $L$ in line 8 of Algorithm \ref{alg:dec}. If a decoder $LDC_{j}$ outputs an empty list, for all possible choices in $L_{j-1}$, then we consider another ordered subset of indices, and start the procedure again.

We now prove that for every error vector $\mathbf{e}$ with $wt(\mathbf{e}) \le \tau$ there exists a \emph{good} set of indices
$\{i_1,\ldots , i_s\} \subset \{ 1,\ldots, l\}$ satisfying condition (\ref{eq:condi}). We should repeat the procedure described above, with every ordered subset of indices and collect all the decoded words, in order to be sure that the ``good'' set of indices is considered.

\begin{theorem}\label{THAlgDecPlotSum}
Let $C$ be the matrix-product code $[C_1 \cdots C_s] \cdot A$,
where $C_1 \supset \cdots \supset C_s$ and $A$ is a non-singular
by columns matrix. Let $\mathbf{e}=(e_1,e_2,\ldots,e_l)\in \fq^{ml}$ be an
error vector with $wt(\mathbf{e}) \le \tau$ (see (\ref{TauValue})). Then
there exists an ordered subset $\{i_1,\ldots , i_s\} \subset \{
1,\ldots, l\}$ satisfying $wt(e_{i_j})\le \tau_j $, for all $j \in \{1,\ldots, s\}$.
\end{theorem}
\begin{proof}
We claim that there exists $i_1$ such that $wt(e_{i_1})\le \tau_1$.
Suppose that there is no $i_1\in\{ 1,\ldots, l\}$ with
$wt(e_{i_1})\le \tau_1$, that is, $wt(e_{i}) \geq \tau_1+1$, for all $i=1,\ldots,l$. This implies that
$$
wt(\mathbf{e})=wt(e_1)+\cdots+wt(e_l) \ge l\tau_1 +l>\tau
$$which contradicts our assumption.

Let us assume that the property holds for a subset $\{i_1,\ldots ,
i_{j-1}\} \subset \{ 1,\ldots, l\}$ of size $j-1<s$. We now prove
it holds for a subset of size $j$. Suppose that there is no $i_j$
with $wt(e_{i_j})\le \tau_j$, that is, $wt(e_{i}) > \tau_j+1$, for all $i\in \{1 ,\ldots,l \}\setminus \{ i_1
, \ldots,  i_{j-1}  \}$. This implies that
\begin{eqnarray*}
wt(\mathbf{e}) & \ge & \sum_{k=1}^{j-1} wt(e_{i_k}) + \sum_{k=j}^{l} wt(e_{i_k})\\
& > & \sum_{k=1}^{j-1} wt(e_{i_k}) + (l-j+1)\tau_j + (l-j+1)\\
 & \ge  &  (l-j+1)\tau_j + (l-j+1)> \tau\\
\end{eqnarray*}which contradicts our assumption and the result holds.
\end{proof}

Summarizing, we can now formulate our decoding algorithm for
$C=[C_1 \cdots C_s] \cdot A \subset \fq^{ml}$, where $C_1 \supset
\cdots \supset C_s$ and $A$ is a non-singular by columns matrix,
in procedural form in Algorithm \ref{alg:dec}.

\newcommand{\decc}{\ensuremath{\mbox{\sc List decoding algorithm for $C=[C_1 \cdots C_s] \cdot A$}}}
\begin{algorithm}[h!]
  \caption{\decc}\label{alg:dec}
  \begin{algorithmic}[1]
    \REQUIRE Received word $\mathbf{p} =\mathbf{c} + \mathbf{e}$ with $\mathbf{c} \in C$ and $wt(\mathbf{e}) \le \tau$. $C_1 \supset \cdots \supset C_s$ nested codes and $A$  a non-singular by columns matrix. Decoder $LDC_i$ for code $C_i$, $i = 1,\ldots ,s$.
    \ENSURE List of all codewords that differ from $p$ in at most $\tau$ places.
    \medskip
    \STATE $\mathbf{p'}=\mathbf{p}$; $A'=A$; $Dec = \{  \}$;
    \FOR{ $\{i_1,\ldots,i_s\}\subset \{1,\ldots , l\}$  }
    \STATE $\mathbf{p}=\mathbf{p'}$; $A=A'$; $U'=\{\mathbf{p}\}$;
      \FOR{$j=1,\ldots , s$}
        \STATE $U=U'$; $U'=\{\}$;
        \FOR{$\mathbf{u}$ in $U$}
             \STATE $L= LDC_j (u_{i_j})$;
             \FOR{$\ell$ in $L$}
                  \STATE $tmp=(0,\ldots,0) \in \mathbb{F}^{m}$;
                  \FOR{$k=j+1,\ldots,s$}
                       \STATE $tmp_{i_k}=u_{i_k}- \frac{a_{j,i_k}}{a_{j,i_j}} \ell$;
                  \ENDFOR
                   \STATE $U'=U'\cup \{tmp\}$;

             \ENDFOR
        \ENDFOR
        \IF{$U'=\{\}$}
            \STATE {\it Break} the loop and consider another $i_1,\ldots ,i_s$ in line 2;
        \ENDIF
        \FOR{$k=j+1,\ldots,s$}
            \STATE $\mathrm{column}_{i_k}(A)= \mathrm{column}_{i_k}(A) - \frac{a_{j,i_k}}{a_{j,i_j}} \mathrm{column}_{i_j}(A)$;
        \ENDFOR
      \ENDFOR
        \FOR{$\mathbf{u}$ in $U'$}
           \STATE Obtain $(c_1,\ldots,c_s)$ from $u_{i_1},\ldots,u_{i_s}$;
           \STATE $\mathbf{p} =[c_1 \cdots c_s] \cdot A$; (see (\ref{MatrixCodeword}) and (\ref{VectorCodeword}))
           \IF{$wt(\mathbf{p}-\mathbf{p'}) \le \tau $}
                \STATE $Dec=Dec \cup \{ \mathbf{p} \}$;
           \ENDIF
        \ENDFOR
    \ENDFOR
  \end{algorithmic}
\end{algorithm}

\begin{corollary}\label{cIsInL1}
If $wt(\mathbf{e}) \leq \tau$ then $\mathbf{c}$ is in the list given as output of Algorithm \ref{alg:dec}.
Hence, the algorithm described in this section is a list decoding algorithm with error bound $\tau$, i.e,
$L=\{\mathbf{c}\in C \mid wt(\mathbf{p}-\mathbf{c})\leq \tau\}$.
\end{corollary}
\begin{proof}
By $\ref{THAlgDecPlotSum}$, there exists an ordered subset $\{i_1,\ldots , i_s\} \subset \{
1,\ldots, l\}$ satisfying $wt(e_{i_j})\le \tau_j $, for all $j \in \{1,\ldots, s\}$. Therefore, $c_{i_j}\in L_j$ and $\mathbf{c}=(\sum_{j=1}^s a_{j,1} c_j, \ldots , \sum_{j=1}^s a_{j,l} c_j )\in L$. Furthermore, all the words at distance $\tau$ from the received word are included in the output list as well.
\end{proof}

\begin{example}\label{ex:ejemplo1}

Consider the matrix-product codes with matrix $A$ of the form 
$$
A=\left(\begin{matrix}
1 & 1 \\
0 & 1
\end{matrix}\right),
$$
and  $C_1 \supset C_2$ Reed-Solomon Codes over $\mathbb{F}_{16}$ with parameters $[15,10,6]$
and $[15,4,12]$, respectively. Therefore, the code $C=[C_1C_2]\cdot A$ has parameters $[30,14,12]$.  We have error bounds $\tau_1=3, \tau_2=7$, for $C_1$ and $C_2$ respectively, using the list decoding algorithm in \cite{peter} or in \cite{Lee-Michael} with multiplicity $v=4$ (see next section for further details). Therefore, the error bound for Algorithm \ref{alg:dec} is $\tau=7$. Note that the error correction capability of $C$ with \cite{hlr} is only $t=5$.

\item Let $ \mathbf{c}=(0,0)$ be the sent word and $\mathbf{p}=(\alpha^{2}x+\alpha x^5+\alpha^5 x^6+\alpha^{14} x^{13},\alpha^5 x^2+\alpha^7 x^6+\alpha^8 x^{10})$ the received word, i.e. $wt(\mathbf{e})=7$. 
\begin{itemize}
\item We  consider the ordered set of indices $\{1,2\}$. So, we decode $p_1 =\alpha^{2}x+\alpha x^5+\alpha^5 x^6+\alpha^{14} x^{13}$ with the list decoding algorithm for $C_1$: we obtain $p_1^{2)} =\alpha^{2}x+\alpha x^5+\alpha^5 x^6+\alpha^{14} x^{7}+\alpha^{10} x^{13}+\alpha^{5} x^{14}$.

Then we compute $p_2^{2)}=p_2-p_1^{2)}=\alpha^{2}x+\alpha^5 x^2+\alpha x^5+\alpha^5 x^6+\alpha x^{7}+\alpha^8 x^{10}+\alpha^{10} x^{13}+\alpha^{5} x^{14}$ and decode it with the list decoding algorithm for $C_2$. However, we get an empty list as output and we do not consider any codeword for the final list.

\item We consider now the ordered set of indices $\{2,1 \}$. Therefore, we decode $p_2=\alpha^5 x^2+\alpha^7 x^6+\alpha^8 x^{10}$ with the list decoding algorithm for $C_1$. We obtain as output $p_2^{2)}=0$. Thus, we compute $p_1-p_2^{2)}=p_1$ and we decode it with the list decoding algorithm for $C_2$. We have obtain $0$. Therefore, we deduce that the sent codeword is $(0,0)$.

\end {itemize}
\end{example}

\begin{remark}\label{re:GR}
Matrix-product codes are generalized concatenated codes \cite{Blackmore-Norton}. There is an efficient decoding algorithm for generalized concatenated codes \cite{Dumer}, the cascaded decoding of multilevel concatenations. In \cite{Guruswami-Rudra}, Guruswami and Rudra generalize this algorithm to a list decoding algorithm for generalized concatenated codes. 
 
The latter are defined as follows: consider $s$ outer codes, say $C^j_{out}$ over $\mathbb{F}_{q^{a_j}}$ with parameters $(N,K_j,D_j)$ for $j=0,\ldots,s-1$. Let $C_{out}=C_{out}^0\times\cdots\times C_{out}^{s-1}=\{(c^0,\ldots,c^{s-1})\mid c^{j}\in C_{out}^j, j=0,\ldots,s-1\}$, understanding $c^j\in C_{out}^j$
as a row vector. Thus, a typical element $c\in C_{out}$ is a $s\times N$ matrix, we denote by $c_k$  the $k$-th column of $c$, for $k=0,\ldots,N-1$.

One also considers a inner code $C_{in}$ over $\mathbb{F}_q$ and a
one-to-one  map $\psi$ from $\mathbb{F}_{q^{a_0}}\times\cdots\times\mathbb{F}_{q^{a_{s-1}}}$ to $C_{in}$
that maps $(i_0,\ldots,i_{s-1})\in \mathbb{F}_{q^{a_0}}\times\cdots\times\mathbb{F}_{q^{a_{s-1}}}$ to a codeword $\psi(i_0,\ldots,i_{s-1})$ in $C_{in}$.

A generalized concatenated code $V$ of order $s$ is the set
$$
V=\{\psi(c_0),\ldots,\psi(c_{N-1})\mid (c_0,\ldots,c_{N-1})\in B\}.
$$

Guruswami and Rudras's algorithm \cite{Guruswami-Rudra} works as follows (we do not attempt to write the algorithm, just to describe its main idea): let $R\in \mathcal{M}(a_0+\cdots+a_{s-1}\times N,\mathbb{F}_q)$ be the received word. For $j=1,\ldots, s-1$ consider the code $C_{in}^j$ generated by all the rows of the generator matrix in $C_{in}$ except the first $a_0+\cdots+a_{j-1}$, say $G_{in}^j$, where  $C_{in}^0=C_{in}$. The algorithm assumes the existence of list decodable algorithms for $C_{in}^j$ and list recovery algorithm for $C_{out}^j$. It is also assumed that the list decoding algorithm for $C_{in}^j$ returns a list of messages while the list recovery
algorithm for $C_{out}^j$ returns a list of codewords.

In the first round one applies, for $i=0,\ldots,N-1$,  a list decoding algorithm of $C_{in}^0$ to $\psi(c_i)$ obtaining a list $S_i^0$. For $i=0,\ldots,N-1$ and for each element in $S_i^0$, one can recover a message $\tilde{c_i}$ and project it into the first component, obtaining a list $T_i^0$. Then it applies list recovery algorithm of $C_{out}^0$ to $\{T_i^0\}_i$ obtaining a list $L_0$ which in particular contains $c^0$.

In the second round we proceed as follows, for each $c\in L_0$, $c\in\mathcal{M}(a_0\times N,\mathbb{F}_q)$, consider the matrix $\hat{c}\in\mathcal{M}(a_0+\cdots+a_{s-1}\times N,\mathbb{F}_q)$ which is the matrix of zeroes with $c$  in the first $a_0$ rows. One may subtract $R=R-(G_{in}^0)^t\hat{c}$, i.e., if $c=c^0$ we are cancelling  it from the received word. Thus we have a new received word $R$ that should be decode with the generalized concatenated code with outer codes $C_{out}^1,\ldots,C_{out}^{s-1}$ and inner code $C_{in}^1$. Since the number of outer codes has dropped by one, repeating this process $s$-times one can successfully list decode the original generalized concatenated code.

Let $C_{in}$ be a linear code over $\mathbb{F}_q$ with parameters $[n,k,d]$ and generator matrix $A$, and  let $C^j_{out}$ be a linear code over $\mathbb{F}_q$ with parameters $[N,K_j,D_j]$ for $j=0,\ldots,k-1$. Consider $\psi$ the encoding linear map of $C_{in}$, i.e., $\psi(i_0,\ldots,i_{k-1})=(i_0,\ldots,i_{k-1})A$. Then the generalized concatenated code $V$ is equal to the matrix product code $[C^1_{out},\ldots,C^s_{out}]\cdot A$. Since a matrix-product code $[C_1,\ldots,C_s]\cdot A$, with $A$ non-singular by columns, is a generalized concatenated code   one might use the algorithm in \cite{Guruswami-Rudra} for matrix-product codes as well.  However, the algorithm in \cite{Guruswami-Rudra} can not be successfully applied for matrix-product codes because the inner code has generator matrix $A$, that is a small matrix (in practice).
\end{remark}

Theorem \ref{THAlgDecPlotSum} guarantees the existence of a \emph{good} set of indices, i.e., 
 satisfying $wt(e_{i_j})\le \tau_j $, for all $j \in \{1,\ldots, s\}$. Hence, in the worst case, we may have to consider  $s!\binom{\ell}{s}$ iterations. However, in average we will consider much fewer iterations. Given a fix set of ordered indices $S=\{i_1,\ldots,i_s\}$, we will estimate for how many error patterns of weight $\tau$, one has a  \emph{good} set of indices. In other words, what is the probability that a set of indices $\{i_1,\ldots,i_s\}$ is a \emph{good} in the worst case, i.e. when $\tau$ errors occur.

\begin{proposition}
If $\tau$ errors occur the probability that a fix set of indices $\{i_1,\ldots,i_s\}$ verifies that 
$wt(e_{i_j}) \le \tau_j ~ \mathrm{for~all}~j \in \{ 1,\ldots, s\}$  is:
\small
$$
\frac{\sum_{a_1=0}^{\tau_1}\sum_{a_2=0}^{min\{\tau-a_1,\tau_2\}}\cdots\sum_{a_{s-1}=0}^{min\{\tau-\sum_{j=1}^{s-2}a_j,\tau_{s-1}\}}  \binom{m}{a_1}\binom{m}{a_2}\cdots\binom{m}{a_{s-1}}\binom{m(\ell-s+1)}{\tau-\sum_{j=1}^{s-1}a_j}}{\binom{m\ell}{\tau}}
$$
\end{proposition}
\begin{proof}
Let us compute the different error vectors $\textbf{e}$ with weight $\tau$ that allow us to have a \emph{good} set of indices. If $wt(e_{i_1}) \le \tau_1$, then we may have $0\leq a_1\leq \tau_1$ errors in block $i_1$, and there are  $\binom{m}{a_1}$ possibilities for having $a_1$ errors in block $i_1$. Assuming that $a_1$ errors occurred in the block $i_1$, we may have  $wt(e_{i_2}) \le \tau_2$ if and only if there are $a_2$ errors in the block $i_2$, where $0\leq a_2\leq min\{\tau-a_1,\tau_2\}$, since $wt(\textbf{e})=\tau$. Hence, there are $\binom{m}{a_2}$ different possibilities for having $a_2$ errors in the second block. Repeating this argument for the first $s-1$ blocks, we may have $\tau-\sum_{j=1}^{s-1}a_j$ errors in the  $\ell-s+1$ remaining blocks (including block $i_s$) and therefore there are  $\binom{m(\ell-s+1)}{\tau-\sum_{j=1}^{s-1}a_j}$ possibilities for the remaining blocks. Overall we have $$\sum_{a_1=0}^{\tau_1}\sum_{a_2=0}^{min\{\tau-a_1,\tau_2\}}\cdots\sum_{a_{s-1}=0}^{min\{\tau-\sum_{j=1}^{s-2}a_j,\tau_{s-1}\}}  \binom{m}{a_1}\binom{m}{a_2}\cdots\binom{m}{a_{s-1}}\binom{m(\ell-s+1)}{\tau-\sum_{j=1}^{s-1}a_j}$$
error patterns that make  $\{i_1,\ldots,i_s\}$ a \emph{good} set of indices. The result holds since there are $\binom{m\ell}{\tau}$ error vectors of weight $\tau$. 	
\end{proof}

In Example \ref{ex:ejemplo1} we have $s=l=2$ and $\tau_1=3$.  Therefore the probability that either $\{1,2\}$ or $\{2,1\}$
is a good set of indices is
$$
\frac{\binom{15}{0}\binom{15}{7}+\binom{15}{1}\binom{15}{6}+\binom{15}{2}\binom{15}{5}
+\binom{15}{3}\binom{15}{4}}{\binom{30}{7}}=\frac{1}{2}.
$$

Finally we consider the complexity of alogithm \ref{alg:dec}.

\begin{theorem}\label{th:complexity}
Let $LDC_1,\ldots,LDC_s$ be the list decoding algorithms considered in Algorithm  \ref{alg:dec} with error bounds $\tau_1,\ldots,\tau_s$, respectively, and that output a list with size bounded by $D_1, \ldots, D_s$. We denote by $R_i$ the complexity of the algorithm $LDC_i$. Then, algorithm \ref{alg:dec} has complexity $$O\left(s!\binom{\ell}{s} (D_1 + \sum_{i=2}^s (\prod_{j=1}^{i-1} D_j) R_i)	  \right).$$
\end{theorem}
\begin{proof}
In the worst case, one should  consider every ordered set of $s$ elements within the $l$ possible indices, that is $s!\binom{\ell}{s}$  ordered sets. For a fix ordered set, we run $LDC_1$ which yields a list of size at most $D_1$, in the worst case. For each element in this list we run $LDC_2$ producing a list of size $D_2$, in the worst case. Hence we will have $D_1 D_2$ words that should be decoded with $LDC_3$. Repeating this process, in the worst case, we will decode $D_1 \cdots D_{s-1}$ words with $LDC_{s}$. 
\end{proof}

\section{List decoding of Matrix-product codes from Reed-Solomon codes with small $s$}\label{se:rs}
The previous algorithm can become computationally intensive as the number of blocks $l$,
the number of blocks that we may need to decode at each iteration $s$ and the error bounds $\tau_1,
\ldots , \tau_s$, increase. Therefore, there are two interesting situations:
considering few blocks and considering codes and error bounds such that there is a small
probability of getting a list with more that one element as output of the list decoding algorithms $LDC_i$.

Let us consider that the constituent codes $C_1 \supset \cdots \supset C_s$ are Reed-Solomon. This family of codes is especially interesting in this setting for two reasons, consider an $[m,k,d]$ Reed-Solomon code, there is an efficient list-decoding algorithm \cite{Guruswami-Sudan} for decoding up to $\tau^{v}$ errors, with multiplicity $v \in \mathbb{N}$, which is computed as follows

 $$\tau^{v} = m - \left \lfloor \frac{l_v}{v} \right \rfloor -1,\mathrm{~where}$$ $$l_v= \left \lfloor \frac{m\binom{v+1}{2}}{r_v} + \frac{(r_v -1)(k-1)}{2} \right \rfloor\mathrm{~and}~r_v\mathrm{~is~calculated~so~that}$$  $$\binom{r_v}{2} \le \frac{m\binom{v+1}{2}}{k-1} < \binom{r_{v} +1}{2}.$$

In particular, one has the algorithm in \cite{Lee-Michael} with complexity $O(D^4vm^2)$ and the one in \cite{peter} with complexity $O(D^4  v m \log^2 m \log \log m)$, where $D$ is the list size, $v$ the multiplicity and  $m$  the length of the code. Notice that, by fixing $v$, we fix $\tau^{v}$ and  bound the list size $D\leq l_v/(k-1)$. Thus, one may obtain the complexity of algorithm \ref{alg:dec} as a function of the multiplicities and the length of the constituent Reed-Solomon codes by Theorem \ref{th:complexity}.

Furthermore, one has a bound for the the probability $p_{\tau^{v}} (C)$ that the output of Guruswami-Sudan's algorithm for the code $C$ with error bound $\tau^{v}$ has more than $1$ codeword \cite{Tom}, given that at most $\tau^{v}$ errors have occurred and assuming that all the error patterns have the same probability. It turns out that this probability may be very small in practice, for example a Reed-Solomon code over $\mathbb{F}_{2^6}$ with parameters $[64,20,45]$ and $\tau^{1}=23$, this probability is $10^{-25}$.

We consider codes and error bounds in such a way that $p_{\tau_i} (C_i)$ is small, we abbreviate $\tau_i^{v_i}$ to $\tau_i$. With the notation of the previous section, consider a received word $\mathbf{p} = \mathbf{c} + \mathbf{e}$ where $\mathbf{e}$ is the error vector with $wt(\mathbf{e}) \le \tau$. Consider the ordered set of indices $\{i_1, \ldots , i_s\} \subset \{1, \ldots , l\}$, if $wt(e_{i_j}) \le \tau_j$ for every $j$ then we say that the set of indices is \emph{good} (otherwise we say that it is \emph{bad}). For a  \emph{good} set of of indices, the sent word $\mathbf{c}$ is in the output list by Theorem \ref{THAlgDecPlotSum}. Furthermore, we claim that with a high probability the output list just contains this word: $LDC_1 (p_{i_1})$ is going to give as output a list containing $\sum_{j=1}^s a_{j,1} c_j$, in Algorithm \ref{alg:dec}. In practice this list will only have one element since the probability of getting just one codeword is $1- p_{\tau_1} (C_1)$. Then we eliminate $c_1$ in the block $i_2$ and decode it using $LDC_2$. Since $w(e_{i_2}) \le \tau_2$ we obtain a list that contains $\sum_{j=2}^s a^{2)}_{j,i_2} c_j$  and with high probability this list has only one codeword. We proceed in the same way for the rest of the blocks and with probability
\begin{equation}\label{fo:prob}\prod_{i=1}^s (1-p_{\tau_i} (C_i))\end{equation} we obtain an output with just one codeword for this set of indices.

Consider now a  \emph{bad} set of indices $\{ i_1, \ldots , i_s \}$, that is, there exists $j$ such that
$w(e_{i_1})<\tau_1, \ldots, w(e_{i_{j-1}})<\tau_{j-1}$, but $w(e_{i_j})<\tau_j$,
then the block $j$ will not be correctly decoded. Again, with probability (\ref{fo:prob}) we will obtain at most one codeword $\mathbf{p}'$ for this set of indices, we do not know anything about this codeword excepting that it is not the sent one. However, we claim that with a high probability the codewords obtained with this \emph{bad} set of indices will be discarded in line 26 of Algorithm \ref{alg:dec}, namely, we claim that $wt(\mathbf{p}-\mathbf{p}')> \tau$ with at least probability $1-lp_{\tau_1} (C_1)$.

\begin{lemma}\label{Lema:bound}
Let $\mathbf{p},\mathbf{p}' \in C=[C_1 \cdots C_s] \cdot A$, with $\mathbf{p} \neq \mathbf{p}'$ and $C_1 \supset \cdots \supset C_s$ Reed-Solomon codes. For $\tau$ as in (\ref{TauValue}), we have$$P(wt(\mathbf{p}-\mathbf{p}')<\tau) \le l p_{\tau_1} (C_1).$$
\end{lemma}
\begin{proof}

If $wt(\mathbf{p}-\mathbf{p}') \le \tau$ then there is $j$ such that$$wt(p_j -p'_j ) \le \tau/l  \le \tau/s \le \tau_1.$$ One has that $P (wt(p_i - p'_i) < \tau_1) = p_{\tau_1} (C_1)$, since  $\mathbf{p} \neq \mathbf{p}'$ and $p_i, p'_i \in C_1$. Thus,$$P(wt(\mathbf{p}-\mathbf{p}')<\tau) \le \sum_{i=1}^l  P (wt(p_i - p'_i) < \tau_1)= l  p_{\tau_1} (C_1)$$since $\mathbf{p} \neq \mathbf{p}'$ and $wt(\mathbf{p}-\mathbf{p}')= \sum_{i=1}^l wt(p_i - p'_i)$.
\end{proof}

An optimal situation arises considering $s=l=2$ and two Reed-Solomon codes $C_1 \supset C_2$
such that for error bounds $\tau_1$ and $\tau_2$, respectively, Guruswami-Sudan's Algorithm outputs a list
of at most 1 element with a high probability (if at most $\tau_1$, $\tau_2$, respectively, errors have occurred). When $l=s=2$, this construction gives the same family of codes as the $(u,u+v)$-construction,
for instance Reed-Muller codes are obtained in that way.

\section{Bounded Distance Decoding of Quasi-Cyclic Codes}\label{se:units}

Let  $C_1, \ldots, C_s \subset \mathbb{F}_q^m$ be cyclic codes of length $m$ and  $A=(a_{i,j})$ an $s\times l$-matrix, with $s\leq l$, whose entries are units in the ring $\fq[x]/(x^m -1)$ or zero. A unit in $\fq[x]/(x^m -1)$ is a polynomial of degree lower than $m$ whose greatest common divisor with $x^m -1$ is $1$. The so-called \emph{matrix-product code with polynomial units} is the set $C=[C_1 \cdots C_s] \cdot A$ of all matrix-products $[c_1 \cdots c_s] \cdot A$ where $c_i\in C_i \subset  \fq[x]/(x^m -1)$  for $i=1,\ldots, s$. These codes were introduced in \cite{hr}.

We consider always a special set of matrices $A$ to be defined below with full-rank over $\mathbb{F}_q[x]/(x^m-1)$. Let $C_i$ with parameters $[m,k_i,d_i]$, then the matrix-product code with polynomial units $C=[C_1 \cdots C_s] \cdot A$ has length $lm$ and dimension
$k=k_1+\cdots+k_s$.

Let $R_i= (a_{i,1},\ldots,a_{i,l})$ be the element of $(\fq [x]/(x^m -1))^l$ consisting of the $i$-th row of $A$, where $i=1,\ldots,s$.
  Let $C_{R_i}$, be the $\fq [x]/(x^m -1)$-submodule of $(\fq [x]/(x^m -1))^l$ generated by $R_1,\ldots, R_i$.
  In other words, $C_{R_i}$ is a linear code over a ring, and we denote by $D_i$ the minimum Hamming weight of the words of $C_{R_i}$, $D_i = \min \{ wt(x) \mid x \in C_{R_i} \}$. In \cite{hr} the following bound on the minimum distance was obtained

\begin{equation}\label{distancia2}
d(C)\geq d^\ast= \min\{d_1D_1,d_2D_2, \ldots ,d_s D_s\}.
\end{equation}

One of the differences between matrix-product codes and matrix-product codes with polynomial units is that the lower bound $d^\ast$  is not sharp for the latter class of codes.

The minimum distance can actually be much larger than $d^\ast$ and several codes with very good parameters were obtained in this way in \cite{hr}. We will provide a bounded distance decoding algorithm for these codes, using the list-decoding Algorithm \ref{alg:dec}.

Matrix-product codes with polynomial units are quasi-cyclic codes \cite{LF} of length $ml$. Although this family provides codes with very good parameters there are no general fast algorithms for decoding them.
The algorithm in \cite{hlr} for matrix-product codes may be used for these family of codes under certain hypothesis, but in that case it only corrects up to $\lfloor \frac{d^\ast -1}{2} \rfloor$.

We remark that the units of a ring form a multiplicative group, however they do not form an additive group. That is, if $f, g \in  \fq[x]/(x^m -1)$ are units, then $fg$ is a unit but $f+g$ or $f-g$ are not a unit in general. This phenomena will impose further restrictions for the list-decoding algorithm since we cannot divide by a non-unit.

\begin{definition}\label{de:uc}
Let $A$ be a $s\times l$ matrix, whose entries are units in the ring $\fq[x]/(x^m -1)$ or zero.
 Let $A_t$ be the matrix consisting of the first $t$ rows of $A$. For $1\leq j_1< \cdots < j_t\leq l$, we denote by $A(j_1,\ldots,j_t)$ the $t\times t$ matrix consisting of the columns $j_1,\ldots,j_t$ of $A_t$.

A matrix $A$ is \textbf{unit by columns} if the determinant of $A(j_1,\ldots,j_t)$ is a unit in $\fq[x]/(x^m -1)$ for each $1\leq t\leq s$ and $1\leq j_1< \cdots < j_t\leq l$. 
\end{definition}

In particular, a unit by column matrix is a non-singular by columns matrix.

Let $C$ be a matrix-product code with polynomial units i.e. $C=[C_1 \cdots C_s] \cdot A$, where Let $C_1 \supset \cdots \supset C_s$ and $A$ is a unit by columns matrix, in particular the elements of the first row of $A$ are non-zero. \\
 
 With the notation of section \ref{BigTau}, consider a received word $\mathbf{p} = \mathbf{c} + \mathbf{e}$ where $\mathbf{e}$ is the error vector with $wt(\mathbf{e}) \le \tau$.

For $s=1$, the definitions  of non-singular by column- and unit by column - matrix are the same. Namely, we can use Algorithm \ref{alg:dec} without any modifications: for $\{i_1\} \subset \{1, \ldots ,l \}$ a \emph{good} set of indices, we decode the block $p_{i_1}$ with $LDC_1$ because the cyclic codes generated by $f$ and by $f  u$, with $f \mid x^m -1$ and $\gcd (u,x^m-1)=1$, are the same code. Then we divide by $a_{1,i_1}$ to recover $c_1$ (in line 24), we can consider the inverse of $a_{1,i_1}$ since the entries of $A$ are units. Actually, this algorithm for $s=1$ is the list decoding version of the algorithm in \cite{Lally} for 1-generator 1-level quasi-cyclic codes.

For $s \ge 2$, for each elimination step in Algorithm \ref{alg:dec}, we are dividing by $a_{j,i_j}$ (in lines 11, 20), we claim that this can be performed because $a_{j,i_j}$ is a unit. Let $A^{k)}$ denote the matrix obtained recursively from $A$ by performing the following $l-(k-1)$ elementary column operations (see section \ref{BigTau}):
$$\mathrm{column}_{i}(A^{k)})= \mathrm{column}_{i}(A^{k-1)}) -
\frac{a^{k-1)}_{k-1,i}}{a^{k-1)}_{k-1,i_{k-1}}}\mathrm{column}_{i_{k-1}}(A^{k-1)}),$$
for each $i \notin \{ i_1,\ldots , i_{k-1} \}$. These operations introduce $l-(k-1)$ additional zero elements in the $k-1$-th row of $A^{k)}$ at each iteration. Hence the submatrix of $A^{k)}$ given by the first $k$ rows and the $i_1, \ldots , i_{k}$ columns, is a triangular matrix (in this case, a column permutation of a lower triangular matrix) whose determinant is $a^{k)}_{1,i_1} \cdots a^{k)}_{k,i_{k}}$. Since $A$ is unit by columns, this minor is a unit. Hence, $a^{k)}_{k,i_{k}}$ is a unit, since the units form a multiplicative group.

Thus, we have a list-decoding algorithm with error bound $\tau$ as in (\ref{TauValue}). Furthermore, we can use it also for unique decoding up to the capacity of the code if $\tau = \lfloor \frac{d(C) -1}{2} \rfloor$.

\begin{theorem}
Consider a matrix-product code with polynomial units $C=[C_1 \cdots C_s] \cdot A$, where $C_1 \supset \cdots \supset C_s$ and $A$ is a unit by columns matrix. Let $\tau =  \lfloor \frac{d(C) -1}{2} \rfloor$, then the list decoding Algorithm \ref{alg:dec} is a unique decoding algorithm for $C$.
\end{theorem}
\begin{proof}
We have seen above that Algorithm \ref{alg:dec} can be successfully applied in this setting. Hence, this algorithm is a list decoding algorithm, by Corollary \ref{cIsInL1}. In particular, if $wt(\mathbf{e}) \le \tau$ the sent word is in the output list. Moreover, since $\tau =  \lfloor \frac{d(C) -1}{2} \rfloor$ there is no other codeword at distance $\tau$ from the received word and the result holds.
\end{proof}

\begin{example}

Let $s=1$, $l=2$, and let $C_1$ be the Reed-Solomon code with parameters $[15,8,8]$ and generator polynomial $f=x^7 + \alpha^6x^6 + \alpha^{13}x^5 + \alpha^{12}x^4 + \alpha x^3 + \alpha^{10}x^2 + \alpha^{11}x + \alpha^{13}$, where $\alpha$ is a primitive root $\mathbb{F}_{16}$. Let $C=[C_1]\cdot A$, where $A=[1,x^4 + \alpha^5x^3 + \alpha x^2 + \alpha^{11}x + \alpha^{14}]$, with $\alpha\in\mathbb{F}_{16}$ a primitive element. One has that $C$ is a quasi-cyclic code with parameters $[30,8,19]$. 

The error correction capability of $C$ is $t=9$. However, with the algorithm in \cite{hlr} we can only decode up to $7$ errors, since $d^\ast=16$. Considering a list decoding algorithm with multiplicity $2$ for $C_1$, we have an error bound $\tau_1=4$. Hence the error correction capability of Algorithm \ref{alg:dec} is $\tau=2\tau_1+1=9$ and we have that it is a unique decoding algorithm for $C$.
\end{example}

\begin{example}

Let $s=1$, $l=2$, and let $C_1$ be the Reed-Solomon code with parameters $[15,5,11]$ and generator polynomial $f=x^{10} + \alpha^2 x^9 + \alpha^3 x^8 + \alpha^9 x^7 + \alpha^6 x^6 + \alpha^{14} x^5 + \alpha^2 x^4 + \alpha x^3 +
\alpha^6 x^2 + \alpha x + \alpha^{10}$, where $\alpha$ is a primitive element in $\mathbb{F}_{16}$. Let $C=[C_1]\cdot A$, where $A=[1,x^3 + \alpha^3 x^2 + \alpha^{14} x + \alpha^9]$. One has that $C$ is a quasi-cyclic code with parameters $[30,5,24]$, which is the best known code in \cite{MinT}.

The error correction capability of $C$ is $t=11$, however, with the algorithm in \cite{hlr} we can only decode up to $10$ errors, since $d^\ast=22$. Considering a list decoding algorithm for $C_1$ with multiplicity $v=1$, we have error bound $\tau_1=5$  and Algorithm \ref{alg:dec} decodes up to the half of the minimum distance since its correction capability is $\tau=2\tau_1+1=11$. However, considering a list decoding algorithm for $C_1$ with multiplicity $v=8$, we have error bound $\tau_1=7$. Hence, the error bound for Algorithm \ref{alg:dec} is $\tau=2\tau_1+1=15$, which is a list decoding algorithm for $C$.
\end{example}

\begin{example}
Let $s=l=2$ and consider a matrix $A$ of the form 
$$
A=\left(\begin{matrix}
1 & g \\
0 & 1
\end{matrix}\right),
$$
with $g$ a unit in $\mathbb{F}_2[x]/(x^m -1)$. One has that $A$ is a unit by column matrix.

Consider $C_1 \supset C_2$ Reed-Solomon Codes over $\mathbb{F}_{16}$ with parameters $[15,13,3]$
and $[15,8,8]$, respectively. We consider the unit $g=x^5 + \alpha^{10}x^3 + \alpha^2x^2 + \alpha^2$. One has that the code $C=[C_1C_2]\cdot A$ has parameters $[30,21,7]$. Hence, its error correction capability is $t=3$. Let $\tau_1=1, \tau_2=3$ be the error bounds for $C_1$ and $C_2$, for list decoding algorithms with multiplicity $1$. Thus, the error bound for Algorithm \ref{alg:dec} is $\tau=3$ and we have a unique decoding algorithm for $C$.

Note that we may consider unique decoding algorithms for $C_1$ and $C_2$ since $\tau_i=\lfloor \frac{d_i -1}{2} \rfloor$, for $i=1,2$. Hence, in this case, we can reduce the complexity of the algorithm by considering unique decoding algorithms.
\end{example}

\section{Conclusion}

In this article we described a list-decoding algorithm for a class of Matrix-Product codes, we computed its error bound and complexity. This algorithm can become computationally intense, however we show that for small $s,\ell$ and considering Reed-Solomon codes as constituent codes, the algorithm does not become computationally intense. Furthermore, we are able to bound the probability of getting more than one codeword as output. The main advantage of this approach with respect to Reed-Solomon codes is the possibility of considering longer codes without increasing the field size and still using the fast decoding algorithms \cite{peter,Lee-Michael} for the constituent codes. Moreover, we can consider a bounded distance decoding, that decodes up to half of the minimum distance, for Matrix-Product codes with polynomial units, a family with very good parameters.

\bibliographystyle{plain}
\bibliography{ListD-MPC}

\end{document}